    \documentclass[11pt]{article}%
    \usepackage[in]{fullpage}
    \def\UseBibLatex{1}

\makeatletter
\def\input@path{{styles/}{../styles/}}
\makeatother

\providecommand{\BibLatexMode}[1]{}
\providecommand{\BibTexMode}[1]{#1}

\ifx\UseBibLatex\undefined%
  \renewcommand{\BibLatexMode}[1]{}
  \renewcommand{\BibTexMode}[1]{#1}
\else
  \renewcommand{\BibLatexMode}[1]{#1}
  \renewcommand{\BibTexMode}[1]{}
\fi

\BibLatexMode{%
   \usepackage{sariel_biblatex}%
   \usepackage[english]{babel}%
   \usepackage{csquotes}
   \renewcommand*{\multicitedelim}{\addcomma\space}%
}

\usepackage{xcolor}%

\usepackage{url, graphicx}%
\usepackage{amsmath}%
\usepackage{amssymb}%

   \usepackage[amsmath,thmmarks]{ntheorem}%
   \theoremseparator{.}%

\usepackage{mathcalb}%

\usepackage{mleftright}
\usepackage{wasysym}%
\usepackage{xcolor}%
\usepackage{xspace}%
\usepackage{marvosym}%
\usepackage{euscript}%
\usepackage{stmaryrd}%

\usepackage{titlesec}%
\titlelabel{\thetitle. }%

\definecolor[named]{Blue}{cmyk}{1,0.1,0,0.1}
\definecolor[named]{Yellow}{cmyk}{0,0.16,1,0}
\definecolor[named]{Orange}{HTML}{EC9933}
\definecolor[named]{Red}{HTML}{FF4D4D}
\definecolor[named]{DarkRed}{HTML}{903838}
\definecolor[named]{LightBlue}{cmyk}{0.49,0.01,0,0}
\definecolor[named]{Green}{HTML}{248a6d}
\definecolor[named]{Purple}{cmyk}{0.55,1,0,0.15}
\definecolor[named]{DarkBlue}{HTML}{3D7391}
\definecolor[named]{LightGrey}{HTML}{829197}

\usepackage{hyperref}%
\hypersetup{%
  breaklinks,%
  colorlinks=true,%
  urlcolor=DarkRed,
  linkcolor=DarkRed,
  citecolor=DarkBlue
}

\usepackage[inline]{enumitem}

\newlist{compactenumA}{enumerate}{5}%
\setlist[compactenumA]{topsep=0pt,itemsep=-1ex,partopsep=1ex,parsep=1ex,%
   label=(\Alph*)}%

\newlist{compactenuma}{enumerate}{5}%
\setlist[compactenuma]{topsep=0pt,itemsep=-1ex,partopsep=1ex,parsep=1ex,%
   label=(\alph*)}%

\newlist{compactenumI}{enumerate}{5}%
\setlist[compactenumI]{topsep=0pt,itemsep=-1ex,partopsep=1ex,parsep=1ex,%
   label=(\Roman*)}%

\newlist{compactenumi}{enumerate}{5}%
\setlist[compactenumi]{topsep=0pt,itemsep=-1ex,partopsep=1ex,parsep=1ex,%
   label=(\roman*)}%

\newlist{compactitem}{itemize}{5}%
\setlist[compactitem]{topsep=0pt,itemsep=-1ex,partopsep=1ex,parsep=1ex}%

\newlist{ProblemList}{enumerate}{5}%
\setlist[ProblemList]{topsep=5pt,%
   partopsep=0.03ex,parsep=0.365ex,%
   resume=problemlist,leftmargin=0.6in,%
   label=\bf\color{DarkRed}{Prob. \arabic*}:, ref=\arabic*}

\newcommand{\HLinkShort}[2]{\hyperref[#2]{#1\ref*{#2}}}
\newcommand{\HLink}[2]{\hyperref[#2]{#1~\ref*{#2}}}
\newcommand{\HLinkPage}[2]{\hyperref[#2]{#1~\ref*{#2}%
      $_\text{p\pageref{#2}}$}}
\newcommand{\HLinkPageOnly}[1]{\hyperref[#1]{Page~\refpage*{#1}%
      $_\text{p\pageref{#1}}$}}

\newcommand{\HLinkSuffix}[3]{\hyperref[#2]{#1\ref*{#2}{#3}}}
\newcommand{\HLinkPageSuffix}[3]{\hyperref[#2]{#1\ref*{#2}%
      #3$_\text{p\pageref{#2}}$}}

\newcommand{\figlab}[1]{\label{fig:#1}}
\newcommand{\figref}[1]{\HLink{Figure}{fig:#1}}

\newcommand{\seclab}[1]{\label{sec:#1}}

\providecommand{\deflab}[1]{\label{def:#1}}
\newcommand{\defref}[1]{\HLink{Definition}{def:#1}}

\newcommand{\apndlab}[1]{\label{apnd:#1}}
\newcommand{\apndref}[1]{\HLink{Appendix}{apnd:#1}}

\newcommand{\defrefY}[2]{\hyperref[def:#2]{#1}}

\newcommand{\lemlab}[1]{\label{lemma:#1}}
\newcommand{\lemref}[1]{\HLink{Lemma}{lemma:#1}}%

\providecommand{\algref}[1]{}%
\renewcommand{\algref}[1]{\HLink{Algorithm}{Algorithm:#1}}%

\newcommand{\thmlab}[1]{{\label{theo:#1}}}
\newcommand{\thmref}[1]{\HLink{Theorem}{theo:#1}}

\providecommand{\eqlab}[1]{}%
\renewcommand{\eqlab}[1]{\label{equation:#1}}
\newcommand{\Eqref}[1]{\HLinkSuffix{Eq.~(}{equation:#1}{)}}

\definecolor{blue25}{rgb}{0,0,0.7}

\providecommand{\emphic}[2]{%
   \textcolor{blue25}{%
      \textbf{\textup{\emph{#1}}}}%
   \index{#2}}

\providecommand{\emphi}[1]{\emphic{#1}{#1}}

   \numberwithin{figure}{section}%
   \numberwithin{table}{section}%
   \numberwithin{equation}{section}%

\newcommand{\Term}[1]{\textsf{#1}\xspace}%
\providecommand{\Mh}[1]{#1}%

\newcommand{\pth}[1]{\mleft({#1}\mright)}

\providecommand{\set}[1]{\mleft\{#1\mright\}}
\renewcommand{\set}[1]{\mleft\{#1\mright\}}
\newcommand{\Set}[2]{\left\{ #1 \;\middle\vert\; #2 \right\}}

\newcommand{\ceil}[1]{\left\lceil {#1} \right\rceil}
\newcommand{\floor}[1]{\left\lfloor {#1} \right\rfloor}
\newcommand{\etal}{{e{}t~a{}l.}\xspace}

\newcommand{\cardin}[1]{\left\vert {#1}  \right\vert}

\newcommand{\eps}{\varepsilon}
\newcommand{\KS}{\Mh{K}}%

\newcommand{\PS}{\Mh{P}}%

\newcommand{\HCX}[1]{\ensuremath{[0,1]^{#1}}}

\newcommand{\Body}{\Mh{\Xi}}
\newcommand{\BodyC}{\Mh{\mathcal{C}}}%

\newcommand{\volX}[1]{\mathsf{vol}\pth{#1}}

\newcommand{\Ball}{\mathsf{b}}%

\newcommand{\EC}{\Mh{\mathcal{E}}}%

\renewcommand{\flat}{\Mh{\varphi}}%
\newcommand{\AT}{\mathbf{M}}%

\newcommand{\Ground}{{\ensuremath{\mathcal{U}}} }

\newcommand{\Ranges}{{\mathcal{R}}}%
\newcommand{\RangesC}{\Mh{\mathcal{C}}}%
\newcommand{\Samp}{\Mh{\mathsf{S}}}
\newcommand{\VC}{\Term{VC}}%

\newcommand{\RangeSpace}{\Mh{\mathsf{X}}}%
\newcommand{\range}{\nu}%

\newcommand{\Matousek}{Matou{\v{s}}ek\xspace}
\newcommand{\Barany}{B{\'a}r{\'a}n{}y\xspace}%

\providecommand{\TPDF}[2]{\texorpdfstring{#1}{#2}}
\newcommand{\kenet}[2]{$(#1, #2)$-net\xspace}
\newcommand{\knet}[1]{\kenet{#1}{\eps}}
\newcommand{\enet}[1]{\kenet{k}{#1}}
\newcommand{\net}{\kenet{k}{\eps}}
\renewcommand{\th}{th\xspace} \newcommand{\bin}{\Term{bin}}%
\newcommand{\rev}{\Term{rev}}%
\providecommand{\br}{\Term{b{}r}}%
\renewcommand{\br}{\Term{b{}r}}%
\newcommand{\pp}{\Mh{p}}%
 
\renewcommand{\Re}{\mathbb{R}}%
\newcommand{\naive}{naive\xspace}%

\newcommand{\lenX}[1]{\left\|#1 \right\|}

\newcommand{\rv}{\Mh{\mathcalb{r}}}%

\providecommand{\IntRange}[1]{\mleft\llbracket #1 \mright\rrbracket}
\newcommand{\IRX}[1]{\IntRange{#1}}%

   \theoremseparator{.}%

\theoremstyle{plain}%
\newtheorem{theorem}{Theorem}[section]

\newtheorem{lemma}[theorem]{Lemma}

\theoremstyle{plain}%
\theoremheaderfont{\sf} \theorembodyfont{\upshape}%
\newtheorem*{remark:unnumbered}[theorem]{Remark}%
\newtheorem{remark}[theorem]{Remark}%
\newtheorem{definition}[theorem]{Definition}

\newcommand{\myqedsymbol}{\rule{2mm}{2mm}}

\theoremheaderfont{\em}%
\theorembodyfont{\upshape}%
\theoremstyle{nonumberplain}%
\theoremseparator{}%
\theoremsymbol{\myqedsymbol}%
\newtheorem{proof}{Proof:}%

\newlength{\arxivwidth}

\newcommand{\rankX}[1]{\mathsf{rank}\pth{#1}}%

\newcommand{\BB}{\mathcal{B}}

\newcommand{\atgen}{\symbol{'100}}
\newcommand{\SarielThanks}[1]{\thanks{Department of Computer Science;
      University of Illinois; 201 N. Goodwin Avenue; Urbana, IL,
      61801, USA; {\tt sariel\atgen{}illinois.edu}; {\tt
         \url{http://sarielhp.org/}.} #1}}

\newcommand{\MitchellThanks}[1]{%
   \thanks{%
      Department of Computer Science;
      University of Illinois; 201 N. Goodwin Avenue; Urbana, IL,
      61801, USA; {\tt mfjones2\atgen{}illinois.edu}; {\tt
         \url{http://mfjones2.web.engr.illinois.edu/}.} #1}}

\usepackage{caption}%

\BibLatexMode{%
   \bibliography{weak_lines}%
}

\begin{document}

   \title{A Note on Stabbing Convex Bodies with Points, Lines, and
      Flats}%
   \date{\today}

   \author{%
      Sariel Har-Peled%
      \SarielThanks{Work on this paper was partially supported by NSF AF
         award CCF-1907400.}%
      \and%
      Mitchell Jones%
      \MitchellThanks{Work on this paper was partially supported by NSF
         AF award CCF-1907400.}%
   }
   
   \maketitle
   \begin{abstract}
   Consider the problem of constructing weak $\eps$-nets where the
   stabbing elements are lines or $k$-flats instead of points. We
   study this problem in the simplest setting where it is still
   interesting---namely, the uniform measure of volume over the
   hypercube $[0,1]^d\bigr.$. Specifically, a $(k,\eps)$-net is a set
   of $k$-flats, such that any convex body in $[0,1]^d$ of volume
   larger than $\eps$ is stabbed by one of these $k$-flats.  We show
   that for $k \geq 1$, one can construct $(k,\eps)$-nets of size
   $O(1/\eps^{1-k/d})$. We also prove that any such net must have size
   at least $\Omega(1/\eps^{1-k/d})$. As a concrete example, in three
   dimensions all $\eps$-heavy bodies in $[0,1]^3$ can be stabbed by
   $\Theta(1/\eps^{2/3})$ lines.  Note, that these bounds are
   \emph{sublinear} in $1/\eps$, and are thus somewhat surprising.
   The new construction also works for points providing a weak
   $\eps$-net of size $O(\tfrac{1}{\eps}\log^{d-1} \tfrac{1}{\eps} )$.
\end{abstract}

\maketitle

\section{Introduction}

\noindent%
\textbf{Notations.} Throughout, we use $O_d$, $\Omega_d$, and
$\Theta_d$ to hide constants depending on the dimension $d$. We use
$\IRX{n}$ to denote the set $\{1,\ldots, n\}$.

\paragraph*{Range spaces and $\eps$-nets.}  %
A \emphi{range space} is a pair $\RangeSpace = (\Ground, \Ranges)$,
where $\Ground$ is the \emphi{ground set} (finite or infinite) and
$\Ranges$ is a (finite or infinite) family of subsets of $\Ground$.
The elements of $\Ranges$ are \emphi{ranges}.

Suppose that $\Ground$ is a finite set. For a parameter
$\eps \in (0,1)$, a subset $\Samp \subseteq \Ground$ is an
\emphi{$\eps$-net} for the range space $\RangeSpace$, if for every
range $\range \in \Ranges$ with
\begin{math}
    \cardin{\range \cap \Ground} \geq \eps\cardin{\Ground}
\end{math}
has $\range \cap \Samp \neq \varnothing$. The $\eps$-net theorem of
Haussler and Welzl \cite{hw-ensrq-87} implies the existence of
$\eps$-nets of size $O(\delta \eps^{-1} \log \eps^{-1})$, where
$\delta$ is the \VC{} dimension of the range space $\RangeSpace$. The
use of $\eps$-nets is widespread in computational geometry
\cite{m-ldg-02,h-gaa-11}.

\paragraph*{Weak \TPDF{$\eps$}{eps}-nets.} Consider the range space
$(\PS, \RangesC)$, where $\RangesC$ is the collection of all compact
convex bodies in $\Re^d$ and $\PS \subset \Re^d$ is a set of
points. This range space has unbounded \VC{} dimension---the standard
$\eps$-net constructions do not work in this case. The notion of
\emphi{weak $\eps$-nets} bypasses this issue by allowing the net
$\Samp$ to use points outside of $\PS$. Specifically, any convex body
$\Body$ that contains at least $\eps \cardin{\PS}$ points of $\PS$
must contain a point of $\Samp$.  The first construction of weak
$\eps$-nets in the plane was due to \Barany{} \etal \cite{bfl-nhp-90}
of size $O(1/\eps^{1026})$. For all $d \geq 1$, Alon \etal
\cite{abfk-pswnch-92} were the first to construct weak $\eps$-nets in
$\Re^d$ whose size was bounded in terms of $1/\eps$ and $d$.  In 1995,
Chazelle \etal \cite{ceggsw-ibwen-95} improved this bound to
$O \bigl( \eps^{-d} \log^{\zeta(d)} \eps^{-1} \bigr)$, where
$\zeta(d) = O(2^{d}(d-1)!)$.  In 2004, \Matousek and Wagner
\cite{mw-ncwen-04} gave an improved construction of weak $\eps$-nets
of size $O_d(\eps^{-d} \log^{f(d)} \eps^{-1})$, where
$f(d) = O(d^2 \log d)$.  Recently, Rubin
\cite{r-ibwen-18,r-sbwenhd-21} gave an improved bound, showing the
existence of weak $\eps$-nets of size
$O_d(\eps^{-(d - 0.5 + \alpha)})$ for arbitrarily small $\alpha >
0$. For more detailed history of the problem, see the introduction of
Rubin \cite{r-ibwen-18,r-sbwenhd-21}. As for a lower bound, Bukh \etal
\cite{bmn-lbwensc-09} gave constructions of point sets for which any
weak $\eps$-net must have size
$\Omega_d(\eps^{-1} \log^{d-1} \eps^{-1})$.  Closing this gap remains
a major open problem.  See \cite{mv-eaen-17} for a recent survey of
$\eps$-nets and related concepts.

\paragraph*{\net{}s and uniform measure.}
A natural extension of weak $\eps$-nets is to allow the net $\Samp$ to
contain other geometric objects. Given a collection of $n$ points
$\PS \subset \Re^d$ and a parameter $k$, $0 \leq k < d$, we define a
(weak) \net to be a collection of $k$-flats $\Samp$ such that if
$\Body$ is a convex body containing at least $\eps n$ points of $\PS$,
then there is a $k$-flat in $\Samp$ intersecting $\Body$. Note that
\knet{0}{}s are exactly weak $\eps$-nets.

In general, one would expect that as $k$ increases, the size of the
\net shrinks. For example, a \knet{1} for a collection of points in
$\Re^3$ can be constructed by projecting the points down onto the
$xy$-plane and applying Rubin's construction in the plane to obtain a
weak $\eps$-net $\Samp$ of size $O(\eps^{-3/2 - \alpha})$
\cite{r-ibwen-18}. Lifting $\Samp$ up back into three dimensions
results in a \knet{1} of the same size, which is smaller than the best
known weak $\eps$-net size in $\Re^3$
\cite{mw-ncwen-04,r-ibwen-18,r-sbwenhd-21}.  However, one might expect
that a \knet{1} of even smaller size is possible in $\Re^3$, as this
construction uses a set of parallel lines (i.e., one would expect the
lines in an optimal net to have multiple orientations).

Here, we study an even simpler version of the problem, where the
ground set is the hypercube $\HCX{d}$. In particular, for
$\eps \in (0,1)$ and $0 \leq k < d$, we are interested in computing
the smallest set $\KS$ of $k$-flats, such that if $\Body$ is a convex
body with $\volX{\Body \cap \HCX{d}} \geq \eps$, then there is a
$k$-flat in $\KS$ which intersects $\Body$. In the following, the set
$\KS$ is a \emphi{\net for volume measure}.  We note that $\HCX{d}$
can be replaced with any arbitrary compact convex body in the
definition -- the size of the \net increases by roughly a factor of
$d^{O(d)}$, see \lemref{any-convex}.

\paragraph*{Deterministic and explicit constructions of $\eps$-nets.}
The randomized algorithm for computing $\eps$-nets, implied by the
$\eps$-net theorem, can be derandomized, but the resulting running
time is exponential in the dimension.  These algorithms work by
repeatedly halving the input point set, using deterministic
discrepancy constructions, until the set is of the desired size
\cite{m-gd-99,c-dmr-01}.

It is an open problem to compute $\eps$-nets in deterministic
polynomial time, in the dimension and $1/\eps$, even for special
cases.  Previous such work on explicit efficient constructions of weak
(and regular) $\eps$-nets include:
\begin{compactenumI}
    \smallskip%
    \item \textsf{Axis-parallel boxes}.  Explicit constructions of
    \knet{0}{}s for volume measure for axis-parallel boxes in $\Re^d$,
    and is briefly mentioned in \cite{bmn-lbwensc-09}. In this case,
    one can construct a \knet{0}{} for volume measure of size
    $2^{O(d \log d)}/\eps$ using Van der Corput sets in two
    dimensions, and Halton-Hammersely sets in higher dimensions.
    These constructions are essentially described in \cite{m-gd-99}
    (in the context of low-discrepancy point sets), the minor
    modifications required in the proofs are described in
    \apndref{0-net-boxes}.

    \smallskip%
    \item \textsf{Grid points and axis-parallel boxes.}  Linial \etal
    \cite{llsz-ecshs-97} studied the problem of constructing explicit
    $\eps$-nets for axis-parallel boxes, where the ground set is
    $\IRX{m}^d$, for some integer $m>0$. The net size is
    $\bigl((m \log d)/\eps\bigr)^{O(1)}$, and the construction time is
    $(md/\eps)^{O(1)}$.

    \smallskip%
    \item \textsf{Halfplanes for vertices of the hypercube, and
       hypersphere.} Rabani and Shpilka \cite{rs-ecsen-10} showed that
    for $\{0,1\}^d$, and halfspaces, one can compute an $\eps$-net of
    size $(d/\eps)^{O(1)}$ (where the constant is dimension
    independent).

    For $\eps = \exp\bigl( -O(\sqrt{n})\bigr)$, they also show a
    construction for volume measure on the hypersphere, for
    halfspaces, with a similar upper bound.
\end{compactenumI}

\subsection{Our results \& paper organization}

First, we show that any \net for volume measure must have size
$\Omega_d(1/\eps^{1-k/d})$ (\lemref{lb-k}). Perhaps surprisingly, we
give a relatively simple construction of \net{}s for volume measure of
size $O_d(1/\eps^{1-k/d})$ for $k \geq 1$ (\thmref{ke-nets-opt}). For
$k=0$, when using points, the same construction works, but the net
size increases to
$O_d\bigl(\tfrac{1}{\eps}\log^{d-1}\tfrac{1}{\eps}\bigr)$.  As far as
the authors are aware, this particular problem has not been addressed
before.

Note that for the case of points and volume measure on the hypercube,
it is enough to build a (regular) $O_d(\eps)$-net for ellipsoids (see
\lemref{reduce-to-pts} below).  In particular, applying the known
deterministic algorithm for computing $\eps$-nets
\cite{m-gd-99,c-dmr-01}, it is not clear what the generated $\eps$-net
is, without running this construction algorithm outright (which seems
quite challenging). In contrast, our algorithm enables us to output
the $i$\th point in the computed net in space and time polylogarithmic
in $O(1/\eps)$.

\section{Preliminaries}

\subsection{Formal definition of \TPDF{\net}{(k,epsilon)-net}}

\begin{definition}
    The affine hull of a point set
    $\PS = \{ \pp_1,\ldots, \pp_n \} \subseteq \Re^d$ is the set
    \begin{equation*}
        \Set{\Bigl.\smash{\sum\nolimits_{i} \alpha_i \pp_i}}
        {\forall i\,\,\, \alpha_i
           \in \Re\,  \text{ and }\, \smash{\sum\nolimits_{i}} \alpha_i = 1}.
    \end{equation*}
    For $k=0,\ldots, d-1$, a \emphi{$k$-flat} is the affine hull of a
    set of $k+1$ (affinely independent) points.
\end{definition}

Thus, a $0$-flat is a point and a $1$-flat is a line.

\begin{definition}
    For parameters $\eps \in (0,1)$ and $k \in \{0,1,\ldots, d-1\}$, a
    set $\KS$ of $k$-flats is a \emphi{\net for volume measure} if for
    any convex body $\Body \subseteq \Re^d$ with
    $\volX{\smash{\Body \cap \HCX{d}}\bigr.} \geq \eps$, there exists
    a flat $\flat \in \KS$ such that
    $\flat \cap \Body \neq \emptyset$.
\end{definition}

\subsection{Brunn-Minkowski inequality and unimodal functions}

The $\Body$ be a convex body in $\Re^d$. For a parameter
$\alpha \in \Re$, let $f(\alpha)$ denote the $(d-1)$-dimensional
volume of $\Body$ intersected with the hyperplane $x = \alpha$. The
Brunn-Minkowski inequality \cite{m-ldg-02,h-gaa-11} implies that the
function $g(\alpha) = f(\alpha)^{1/(d-1)}$ is concave (over the range
where it is not zero). In particular, $g$ is \emphi{unimodal}.
Namely, there exists a $\alpha \in \Re$ such that $g$ is
non-decreasing on $(-\infty, \alpha]$ and non-increasing on
$[\alpha, \infty)$.  As such, the function $f$ itself is unimodal.
See \figref{bm-unimodal}.

\begin{figure}[h]
    \centerline{\includegraphics[scale=0.7]{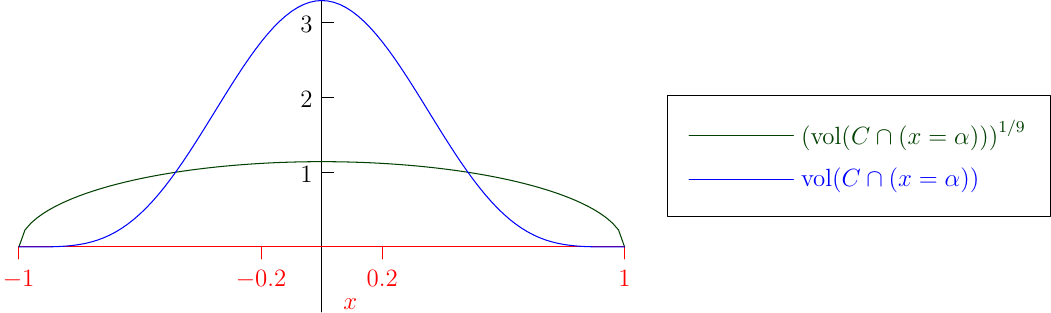}}

    \caption{The slice volume, and its $1/9$\th power, for the unit
       radius ball $\Body$ in $10$ dimensions. This is an example of
       the concavity implied by the Brunn-Minkowski inequality, which
       in turn implies that the slice function is unimodal.}
    \figlab{bm-unimodal}
\end{figure}

\subsection{Approximating convex bodies by ellipsoids}

\subsubsection{Replacing \TPDF{$\HCX{d}$}{[0,1]\^{}d} %
   with other convex bodies}
\lemlab{extensions}

\begin{lemma}
    \lemlab{any-convex}%
    Let $\BodyC$ be an arbitrary compact convex body in $\Re^d$ with
    non-empty interior. Suppose there is a \net for the uniform
    measure on $[0,1]^d$ of size $T(\eps, k, d)$.  For a given integer
    $k < d$ and $\eps \in (0,1)$, there is a collection of $k$-flats $\KS$, of
    size $T(\Omega_d(\eps), k, d))$, such that any convex body $\Body$
    with $\volX{\Body \cap \BodyC} \geq \eps\,\volX{\BodyC}$ is
    intersected by a $k$-flat in $\KS$.
\end{lemma}
\begin{proof}
    Assume without loss of generality that $\Body \subseteq \BodyC$.
    John's ellipsoid theorem \cite{m-ldg-02} implies that there exists
    a non-singular affine transformation $\AT$, and a ball $\Ball$ of
    diameter $1$, such that
    $\Ball/d \subseteq \AT(\BodyC) \subseteq \Ball \subseteq \HCX{d}$,
    where $\Ball/d$ is $\Ball$ scaled by a factor of $1/d$. We have
    that $\volX{\Ball} = c_d 2^{-d}$, where
    $c_d=1/2^{\Theta(d\log d)}$ is the volume of the unit ball in
    $\Re^d$. Additionally,
    \begin{equation*}
        \volX{\smash{\HCX{d}}\bigr.}%
        =%
        1%
        =%
        \frac{2^d}{c_d}\,\volX{\Ball}
        = \frac{(2d)^d}{c_d}\,\volX{\Ball/d}
        \leq \frac{(2d)^d}{c_d}\,\volX{\AT(\BodyC)}.
    \end{equation*}

    Set $\delta = c_d/(2d)^d$. Compute a \enet{\eps'} $\KS$ for
    $\HCX{d}$, where $\eps' = \eps \delta$, which has size
    $T(\eps', k, d)$. We claim that this is a \net with respect to
    $\AT(\Body)$. Indeed, consider any convex body
    $\Body \subseteq \BodyC$ with
    $\volX{\Body \cap \BodyC} \geq \eps\,\volX{\BodyC}$. Since $\AT$
    preserves the ratios of volumes, we have that
    \begin{align*}
      \volX{\smash{\AT(\Body) \cap \HCX{d}} \bigr.}
      &\geq%
        \volX{\AT(\Body) \cap \AT(\BodyC)\bigr.}
        \geq%
        \eps\,\volX{\AT(\BodyC)\bigr.}
      \\&%
      \geq%
      \eps \delta \,\volX{\smash{\HCX{d}}\bigr.}
      =%
      \eps'\,\volX{\smash{\HCX{d}}\bigr.}.
    \end{align*}
    As such, one of the $k$-flats in $\KS$ intersects $\AT(\Body)$.
    After applying the inverse transformation $\AT^{-1}$ to each
    $k$-flat in $\KS$, one of the $k$-flats in $\AT^{-1}(\KS)$
    intersects $\Body$.
\end{proof}

\subsubsection{Its enough to hit ellipsoids}

\begin{lemma}
    \lemlab{reduce-to-pts}%
    Suppose there exists an \knet{k} for the volume measure over
    $\HCX{d}$ for ellipsoids of size $T(\eps,d)$.  Then one can
    construct a \knet{k} for the volume measure over $\HCX{d}$, for
    all convex bodies, of size $T(\eps/d^d, d)$.
\end{lemma}
\begin{proof}
    Consider any convex body $\Body$, such that
    $\volX{\smash{\Body \cap [0,1]^d}\bigr.} \geq \eps$.  Let $\EC$ be
    the ellipsoid of largest volume contained inside
    $\Body \cap [0,1]^d$. By John's ellipsoid theorem, we have that
    $\EC \subseteq \Body \subseteq d\EC$. In particular,
    \begin{equation*}
        \volX{\EC}%
        =%
        \frac{\volX{d\EC}}{d^d}%
        \geq%
        \frac{\volX{\Body}}{d^d}%
        \geq
        \frac{\eps}{d^d}.
    \end{equation*}
    As such, any $(k,\eps/d^d)$-net for ellipsoids is a \knet{k} for
    general convex bodies.
\end{proof}

\section{Lower bound}

\begin{lemma}
    \lemlab{lb-k}%
    For a parameter $\eps \in (0,1)$, and $k \in \{ 0,\ldots, d-1\}$,
    any \net for volume measure over $\HCX{d}$ must have size
    $\Omega_d(1/\eps^{1 - k/d})$.
\end{lemma}
\begin{proof}
    Let $\KS$ be a \net for volume measure.  For each $k$-flat
    $\flat \in \KS$, let $H(\flat,r)$ be the locus of points in
    $\HCX{d}$ within distance at most $r$ from $\flat$ (for $k=1$ in
    three dimensions, this is the intersection of $\HCX{3}$ and the
    cylinder with radius $r$ centered at the line $\flat$). Note that
    a ball $\Ball$ with center $c$ and radius $r$ intersects a
    $k$-flat $\flat$ if and only if $c \in H(\flat, r)$.

    Fix $r = (\eps/\mu)^{1/d}$, where $\mu$ is a constant to be
    determined shortly. We claim that by choosing $\mu$ appropriately,
    if $\KS$ is a \net for volume measure, then the collection of objects
    $\Set{H(\flat, r)}{\flat \in \KS}$ covers $\HCX{d}$. Indeed,
    suppose not. Then there exists a point $\pp \in \HCX{d}$ not
    covered by any of the objects $H(\flat, r)$. This implies that a
    ball $\Ball$ centered at $\pp$ with radius $r$ does not intersect
    any $k$-flat of $\KS$, and its volume is
    $c_d r^{d} = c_d\eps/\mu$, where $c_d$ is a constant that depends
    on $d$. Choose $\mu = c_d$ so that $\Ball$ has volume at least
    $\eps$, but does intersect any $k$-flat of $\KS$. A contradiction
    to the required net property.

    Hence, by the choice of $r$, any \net for volume measure must
    satisfy the condition that
    \begin{equation*}
        \Set{H(\flat, r)}{\flat \in \KS}        
    \end{equation*}
    covers $\HCX{d}$. For any $k$-flat $\flat$, we have
    $\eta = \volX{H(\flat, r)} = O_d(r^{d-k}) = O_d(\eps^{1 - k/d})$.
    Thus, to cover $\HCX{d}$, we have that
    $\cardin{\KS} \geq 1/\eta = \Omega_d(1/\eps^{1 - k/d})$.
\end{proof}

\section{Constructing \TPDF{\net{}s}{nets} for volume measure}
\seclab{ke-nets-opt}

\begin{figure}
    \noindent%
    \includegraphics[page=1,width=0.23\linewidth]{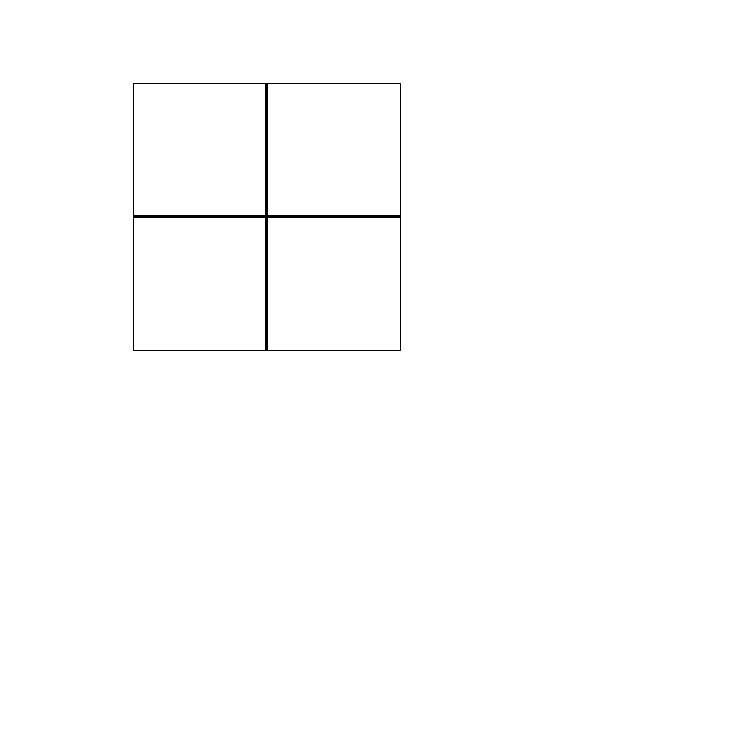}%
    \hfill%
    \includegraphics[page=2,width=0.23\linewidth]{figs/grid_grad}%
    \hfill%
    \includegraphics[page=3,width=0.23\linewidth]{figs/grid_grad}%
    \hfill%
    \includegraphics[page=4,width=0.23\linewidth]{figs/grid_grad}
    \captionof{figure}{The multi-level grid, and its associated
       lines.}%
    \figlab{grid}
\end{figure}

Here, we give a self-contained deterministic and explicit construction
of \net{}s for volume measure for $k \in \{0,\ldots, d-1\}$. The
constructed set size matches the lower bound of \lemref{lb-k} up to
constant factors for $k \geq 1$.

\subsection{Preliminaries}

For a number $x \in (0,1)$, lets its \emphi{rank} be the minimum $i$
such that $2^ix$ is an integer. For example, $\rankX{1/2} =1$, and
$\rankX{7/8} = 3$. Thus, any binary string
$s=s_1 \cdots s_i \in \{0,1\}^i$ of length $i$ that ends in $1$,
corresponds to the number $\sum_{k=1}^i s_k/2^k$ of rank $i$.  Let
\begin{equation*}
    \BB_{t} = \Set{ x / 2^t }{x \in \IRX{2^t-1}}
\end{equation*}
be the set of all numbers in $(0,1)$ of rank at most $t$. Observe that there
are exactly $2^{i-1}$ numbers in $\BB_{t}$ of rank $i$, for
$i=1, \ldots, t$.

\subsection{Construction}
The construction works recursively
on the dimension $d$.

\paragraph*{Base case: $d=1$ and $k=0$.}
Here a \knet{0} for volume measure of size $O(1/\eps)$ follows readily
by spreading $2 + \floor{1/\eps}$ points uniformly on the interval
$\HCX{1}=[0,1]$.

\paragraph*{Base case: $d=k+1$ and $k > 0$.}  Here a \knet{d-1} for
volume measure of size $d/\eps^{1/d} = O(d/\eps^{1 - k/d})$ follows
readily by overlaying a $d$-dimensional grid of size length
$\eps^{1/d}$ over $\HCX{d}$. Each cell in this grid has volume
$\eps$. Thus, the net consists of the hyperplanes forming the
grid\footnote{This requires the convex bodies under consideration to
   be closed.}.

\paragraph*{Induction: $d>k+1$.}
For $i=1,\ldots, d$, and $\varphi \in \BB_{\tau}$, where
\begin{equation}
    \tau = \ceil{\frac{1}{d}\log\frac{1}{\eps}} + 3\ceil{\log (3d)} + 1,
    \eqlab{tau:value}
\end{equation}
consider the hyperplane $h(i,\varphi) \equiv ( x_i = \varphi)$, and
let $\ell = \rankX{\varphi}$. We recursively construct a
\kenet{k}{\eps_\ell} for volume measure on $h(i,\varphi)$ (which lies
in $d-1$ dimensions), where
\begin{equation}
    \eps_\ell = \frac{2^{\ell}\eps}{4d}.
    \eqlab{eps:i}
\end{equation}
Thus, hyperplanes with rank $\ell$ have a finer net on them than
hyperplanes of rank $\ell+1$.  We collect all such $k$-flats built on
all of these hyperplanes of all ranks into a set $\KS$, which is the
desired \net.

See \figref{grid} for an illustration of the construction in two
dimensions.

\paragraph*{Intuition.}  The construction is based on
quadtrees. Starting with the entire cube $\HCX{d}$, we construct $d$
orthogonal hyperplanes which \emph{split} the cube into $2^d$ cubes of
side length $1/2$. We refer to such hyperplanes as \emph{splitting
   hyperplanes}. This splitting process is continued recursively. The
rank of a hyperplane is thus the level of the recursion when it is
being introduced.  All the cubes at the $i$\th level of the
construction have side length $1/2^i$ and they form a grid. The number
of cubes in this grid at the $i$\th level is $2^{di}$.  Observe,
that we recursively construct a net on each ``wall'' of a cell, where
the density of the net is coarser as we go down the recursion.

\subsection{Analysis}

\begin{lemma}
    \lemlab{upper:bound:k:non:zero}%
    For $k \in \IRX{d-1}$, the constructed \net{} for volume measure
    has size at most
    $\beta(d) /\eps^{1 - k/d} = O_d(1/\eps^{1 - k/d})$, where
    $\beta(d) = 2^{O(d-k)} d^{6(d-k-1)+1}$
\end{lemma}
\begin{proof}
    Let $T(\eps,d)$ denote the size of a \net for volume measure over
    $[0,1]^d$ constructed above. The proof is by induction on
    $d$. When $d = k + 1$, we have
    $T(\eps, k+1) \leq (k+1)/\eps^{1/(k+1)}$, by the base case
    described above. So assume $d \geq k+2$ and
    $T(\delta, d') \leq \beta(d')/\delta^{1 - k/d'}$ for all $d' < d$,
    where $\beta(d')$ is a function and $\beta(k+1)=k+1$.  We
    remind the reader that $\eps_i = {2^{i}\eps}/(4d)$ and
    \begin{math}
        \tau \leq \frac{1}{d}\log\frac{1}{\eps} + 3\log d + 8.
    \end{math}
      By the inductive hypothesis, the above construction produces a
    \net of size
    \begin{align*}
      \cardin{\KS}
      &%
        \leq%
        \sum_{i=1}^\tau d2^{i-1}T(\eps_i, d-1)
        \leq%
        d \sum_{i=1}^{\tau} \frac{2^{i-1} \beta(d-1)}{\eps_{i}^{1 - k/(d-1)}}
      \\&
      \leq%
      d \sum_{i=1}^{\tau}
      \frac{2^{i-1} \beta(d-1)}{\pth{2^{i}\eps/(4d)}^{1 - k/(d-1)}}
      \leq%
      \frac{4d^2\beta(d-1)}{\eps^{1 - k/(d-1)}} \sum_{i=1}^{\tau}
      \frac{2^{i-1}}{2^{i - ik/(d-1)}}
      \\&
      =%
      \frac{2d^2\beta(d-1)}{\eps^{1 - k/(d-1)}} \sum_{i=1}^{\tau}
      2^{  ik/(d-1)}%
      \leq%
      \frac{4d^3\beta(d-1)}{\eps^{1 - k/(d-1)}} \cdot 2^{\tau k/(d-1)}
      \\&%
      \leq%
      \frac{4d^3\beta(d-1)}{\eps^{1 - k/(d-1)}} \cdot
      \pth{\frac{256d^3}{\eps^{1/d}}}^{ k/(d-1)}
      \leq%
      \frac{1024d^6\beta(d-1)}{\eps^{1 - k/(d-1) + k/[d(d-1)]}}
      \\&
      =%
      \frac{1024d^6\beta(d-1)}{\eps^{1 - k/d}}.
    \end{align*}
    In particular, we obtain the recurrence
    $\beta(d) = 1024d^6 \beta(d-1)$, which solves to
    $\beta(d) = 2^{O(d-k)} d^{6(d-k-1)+1}$, as $\beta(k+1)=k+1$.
\end{proof}

\begin{lemma}
    For $k =0$, the constructed \knet{0} for volume measure has
    size at most $\frac{\psi(d)}{\eps} \log^{d-1} \frac{1}{\eps}$, where
    $\psi(d) = (\log d)^{O(d^2)}$.
\end{lemma}
\begin{proof}
    We follow the proof of \lemref{upper:bound:k:non:zero}.  Let
    $T(\eps,d)$ denote the size of a \knet{0} for volume measure over
    $[0,1]^d$ constructed above. We have $T(\eps, 1) \leq 3/\eps$.  So
    assume $d \geq 2$, and
    $T(\delta, d') \leq \frac{\psi(d')}{\delta} \log^{d'-1}
    \frac{1}{\delta}$ for all $d' < d$, where $\psi(d')$ is a function
    with $\psi(1)=3$.  As a reminder, we have
    \begin{math}
        \eps_i = {2^{i}\eps}/(4d)
    \end{math}
    and
    \begin{math}
        \tau
        \leq
        10 \log\frac{d}{\eps}.
    \end{math}
    By the inductive hypothesis, the above construction produces a
    \knet{0} for volume measure of size
    \begin{align*}
      \cardin{\KS}
      &%
        \leq%
        d \sum_{i=1}^\tau 2^{i-1}T(\eps_i, d-1)
        \leq%
        d \sum_{i=1}^{\tau} \frac{2^{i-1}
        \psi(d-1)}{\eps_{i}} \log^{d-2} \frac{1}{\eps_i}
      \\&
      \leq%
      d \sum_{i=1}^{\tau} \frac{4d \cdot 2^{i-1}
      \psi(d-1)}{2^{i}\eps} \log^{d-2} \frac{4d}{2^{i}\eps}
      \leq%
      4d^2
      \psi(d-1) \cdot \tau \cdot \frac{1}{\eps} \pth{
      10 \log \frac{d}{\eps}}^{d-2}
      \\&
      \leq%
      4d^2
      (10 \log d)^{d-1}
      \psi(d-1) \cdot \frac{1}{\eps}
      \log^{d-1} \frac{1}{\eps}.
    \end{align*}
    We
    obtain the recurrence
    \begin{math}
        \psi(d) \leq 4d^2 (10 \log d)^{d-1} \psi(d-1),
    \end{math}
    which solves to
    \begin{math}
        \psi(d) = ( \log d)^{O(d^2)}.
    \end{math}
\end{proof}

\begin{lemma}
    The constructed set $\KS$ is a \net{} for volume measure over
    $\HCX{d}$.
\end{lemma}

\begin{proof}
    Let $\Body$ be a convex body contained in $\HCX{d}$ with volume at
    least $\eps$. Assume, for the sake of contradiction, that $\Body$
    is not stabbed by any of the $k$-flats of $\KS$.  The constructed
    set being a net for the base cases of the construction ($d=k+1$ or
    $d=1$ and $k=0$) are immediate.

    So, let $h(\alpha)$ be the hyperplane orthogonal to the first axis
    which intersects the first axis at $\alpha \in \Re$. Define the
    functions
    \begin{equation*}
        f(\alpha) = \volX{\Body \cap h(\alpha)\bigr.}
        \qquad
        \text{and}
        \qquad
        g(\alpha) = f(\alpha)^{1/(d-1)}.
    \end{equation*}
    By the Brunn-Minkowski inequality, the function $g(\alpha)$ is
    concave and unimodal.  Define the point $x^* \in [0,1]$ so that
    $x^\star = \arg\max_{\alpha} f(\alpha)$.

    Let $V(\Delta) = f( x^\star + \Delta)$, and let
    $v(\Delta) = (V(\Delta))^{1/(d-1)}$. The function $v$, being a
    translation of $g$, is concave and unimodal.  Let $\kappa$ be the
    maximum index in $\IRX{\tau}$, such that $\eps_\kappa \leq V(0)$,
    see \Eqref{eps:i}.  Let $\rv_i\geq0$ be the maximum number such
    that $V(\rv_i) = \eps_i$, for $i=1,\ldots, \kappa$. As we can assume
    that $\Body$ is smooth, it is easy to verify the $\rv$s are well
    defined.

    Observe that if
    $\rv_i \geq 1/2^i$, then there is hyperplane orthogonal to the
    first axis that has a recursive construction of a net on it of
    level $i$, for $\eps_i$, that lies in the range
    $[x^\star, x^\star + 1/2^i]$. This by induction would imply that
    the net intersects $\Body$. We thus assume from this point on that
    \begin{equation*}
        \rv_i < \frac{1}{2^i},
    \end{equation*}
    for all $i$.  Observe that
    $\rv_1 \geq \rv_2 \geq \cdots \geq \rv_\kappa$, as
    $\eps_1 < \eps_2 < \cdots < \eps_\kappa$ (more specifically,
    $\eps_i = 2\eps_{i-1}$ for all $i$).

    \begin{figure}[h]
        \centering%
        \includegraphics{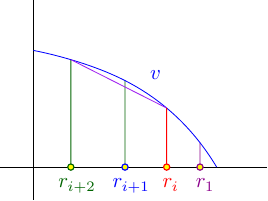}
        \caption{By the choice of
           $\rv_{\kappa} \leq \ldots \leq \rv_1$, we have
           $v(\rv_{\kappa}) \geq \ldots \geq v(\rv_1)$.}
        \figlab{concave}
    \end{figure}

    The concavity of $v(\cdot)$, see \figref{concave}, implies that
    \begin{align*}
      \frac{v(\rv_{i+2}) - v(\rv_{i+1})}
      {\rv_{i+2} - \rv_{i+1}}
      &\geq
        \frac{v(\rv_{i+1}) - v(\rv_{i})}
        {\rv_{i+1} - \rv_{i}}
      \qquad\implies\qquad%
      \frac{\rv_{i+1} - \rv_{i}}
      {\rv_{i+2} - \rv_{i+1}}
      \leq
      \frac{v(\rv_{i+1}) - v(\rv_{i})}
      {v(\rv_{i+2}) - v(\rv_{i+1})},
    \end{align*}
    as $\rv_{i+1} - \rv_{i} < 0$ and
    $v(\rv_{i+2}) - v(\rv_{i+1}) > 0$.  Since
    \begin{math}
        V(\rv_{i+1}) = \eps_{i+1} = 2\eps_i = 2V(\rv_i),
    \end{math}
    we have that
    \begin{math}
        v(\rv_{i+1}) = 2^{1/(d-1)}v(\rv_i).
    \end{math}
    For $i < \kappa$, let $\ell_{i} = \rv_{i} - \rv_{i+1}$. Plugging
    this into the above, observe
    \begin{align*}
      \frac{\ell_{i}}{\ell_{i+1}}
      &=%
        \frac      {\rv_{i} - \rv_{i+1}}
        {\rv_{i+1} - \rv_{i+2}}
        \leq
        \frac{v(\rv_{i+1}) - v(\rv_{i})}
        {v(\rv_{i+2}) - v(\rv_{i+1})}
      =%
      \frac{(2^{1/(d-1)}-1) v(\rv_i)}
      {2^{1/(d-1)}(2^{1/(d-1)}-1 )v(\rv_{i})}
      =%
      \frac{1}{2^{1/(d-1)}}.
    \end{align*}
    Since $\ell_{\kappa-1} \leq \rv_{\kappa-1} \leq 1/2^{\kappa-1}$, we
    have
    \begin{align*}
      \rv_1
      &=%
        \rv_\kappa + \sum_{i=1}^{\kappa-1} \ell_i
        \leq%
        \rv_\kappa +
        \ell_{\kappa-1}\pth{1 + \frac{1}{2^{1/(d-1)}} +
        \frac{1}{2^{2/(d-1)}} + \cdots }
        \\&%
      \leq%
      \rv_\kappa + 2 d \ell_{\kappa-1}
      \leq%
      (2d+1)
      \rv_{\kappa-1}
      \leq%
      \frac{2d+1}{2^{\kappa-1}}
      <
      \frac{\eps^{1/d}}{4d^2},
    \end{align*}
    as $\kappa \leq \tau$, and by the value of $\tau$, see
    \Eqref{tau:value}.

    Let $I_1$ be the maximum interval, where the value of
    $V(x) \geq \eps_1$ for any $x \in I_1$. By the above, we have
    that if the net does not intersect $\Body$, then
    $\lenX{I_1} \leq 2\rv_1 \leq 2{\eps^{1/d}}/(4d^2)$.

    We define $I_2, \ldots, I_d$ in a similar fashion on the other
    axes, and the same argumentation would imply that
    $\lenX{I_j} \leq 2{\eps^{1/d}}/(4d^2)$, for all $j$. Furthermore,
    any plane orthogonal to the axes that avoids the box
    $B = I_1 \times I_2 \cdots \times I_d$ has an intersection with
    $\Body$ of volume at most $\eps_1$. We conclude that the total
    value of $\Body$ is at most
    \begin{equation*}
        \volX{\Body}%
        \leq%
        \volX{B} +
        \sum_{j=1}^d \int_{y \in [0,1] \setminus I_j}
        \volX{\Body \cap (x_j = y)\Bigr.} dy
        \leq
        \prod_{j=1}^d \lenX{I_j} + d \eps_1
        <%
        \eps,
    \end{equation*}
    which is a contradiction to $\volX{ \Body} \geq \eps$.
\end{proof}

Putting the above together, we get our main result.
\begin{theorem}
    \thmlab{ke-nets-opt}%
    Given $\eps \in (0,1)$ and $k \in \{1,\ldots, d-1\}$, the above is
    a deterministic and explicit construction of a \net{} for volume
    measure over $[0,1]^d$ of size
    $\beta(d) /\eps^{1 - k/d} = O_d(1/\eps^{1 - k/d})$, where
    $\beta(d) = 2^{O(d-k)} d^{6(d-k-1)+1}$.

    For $k=0$, the above construction has size
    $\frac{\psi(d)}{\eps} \log^{d-1} \frac{1}{\eps}$, where
    $\psi(d) = (\log d)^{O(d^2)}$.
\end{theorem}

\begin{remark}
    (A) Our upper bound for the case of points matches the lower bound
    $\Omega_d(\frac{1}{\eps} \log^{d-1} \frac{1}{\eps})$ of Bukh \etal
    \cite{bmn-lbwensc-09} (which holds for somewhat different
    settings). This seems to be somewhat coincidental, as the
    $\eps$-net theorem implies, in this case, a smaller weak
    $\eps$-net for volume measure of size
    $O(\tfrac{d}{\eps} \log \frac{1}{\eps})$, via the reduction to
    ellipsoids, see\lemref{reduce-to-pts}.

    \medskip

    (B) The construction here is orthogonal in nature.  For the case
    of $(0,\eps)$-nets, the generated set is significantly larger than
    the Halton-Hammersely set (see \defref{hh} and
    \lemref{0-net-box-d}) which works for axis-aligned boxes. General
    convex bodies do not have the same predictable ``behavior'' of
    axis-aligned boxes, thus maybe explaining the need for a larger
    net.
\end{remark}

\section{Conclusions}

The main open problem left by our work is bounding the size of \net{}s in
the general case. That is, the input is a set $\PS$ of $n$ points in
$\Re^d$, and we would like to compute a minimum set of $k$-flats which
stab all convex bodies containing at least $\eps n$ points of $\PS$.
As noted earlier, there is a \net{} of asymptotically
the same size as of a weak $\eps$-net in $\Re^{d-k}$. This follows by
projecting the point set to a subspace of dimension $d-k$, constructing
a regular weak $\eps$-net, and lifting the net back to the original space.
Can one do better than this somewhat \naive construction?

Note that it is easy to show a lower bound of size $\Omega(1/\eps)$
for \knet{1}{}s in the general case. Take a point set that consists of
$\ceil{2/\eps}$ equally sized clusters of tightly packed points, such
that no line passes through three clusters. Namely, our sublinear
results in $1/\eps$ are special for the uniform measure on the
hypercube.

\paragraph*{Acknowledgements.} %
We thank an anonymous reviewer for sketching an improved construction
of \net{}s for $k \geq 1$, which led to \thmref{ke-nets-opt}.  Our
previous construction had an additional $\log$ term. We also thank the
anonymous reviewers for detailed comments that improved the
paper.

\BibLatexMode{\printbibliography}

\appendix

\section{\TPDF{\knet{0}{}s}{(0,eps)-nets} for %
   axis-aligned boxes}
\apndlab{0-net-boxes}

Here we show the existence of a \knet{0} of size $O(1/\eps)$ that
intersects any axis-aligned box $B$ that has the property that
$\volX{B \cap \HCX{2}} \geq \eps$. The following constructions are
essentially described in \cite{m-gd-99} (in the context of
low-discrepancy point sets), however the proofs use similar tools. We
give the proofs for completeness.

\begin{definition}[the Van der Corput set]
    \deflab{vdc}%
    For an integer $\alpha$, let $\bin(\alpha) \in \set{0,1}^\star$
    denote the binary representation of $\alpha$, and
    $\rev(\bin(\alpha))$ be the reversal of the string of digits in
    $\bin(\alpha)$. We define $\br(\alpha) \in [0,1]$ to be the
    \emphi{bit-reversal} of $\alpha$, which is defined as the number
    obtained by concatenating ``$0.$'' with the string
    $\rev(\bin(\alpha))$. For example, $\br(13) = 0.1011$. Formally,
    if $\alpha = \sum_{i=0}^\infty 2^i b_i$ with $b_i \in \set{0,1}$,
    then $\br(\alpha) = \sum_{i=0}^\infty b_{i}/2^{i+1}$.

    For an integer $n$, the \emphi{Van der Corput set} is the
    collection of points $\pp_0, \ldots, \pp_{n-1}$, where
    $\pp_i = (i/n, \br(i))$. See \figref{vdc}.
\end{definition}

\begin{figure}
    \centering %
    \includegraphics[scale=0.55]{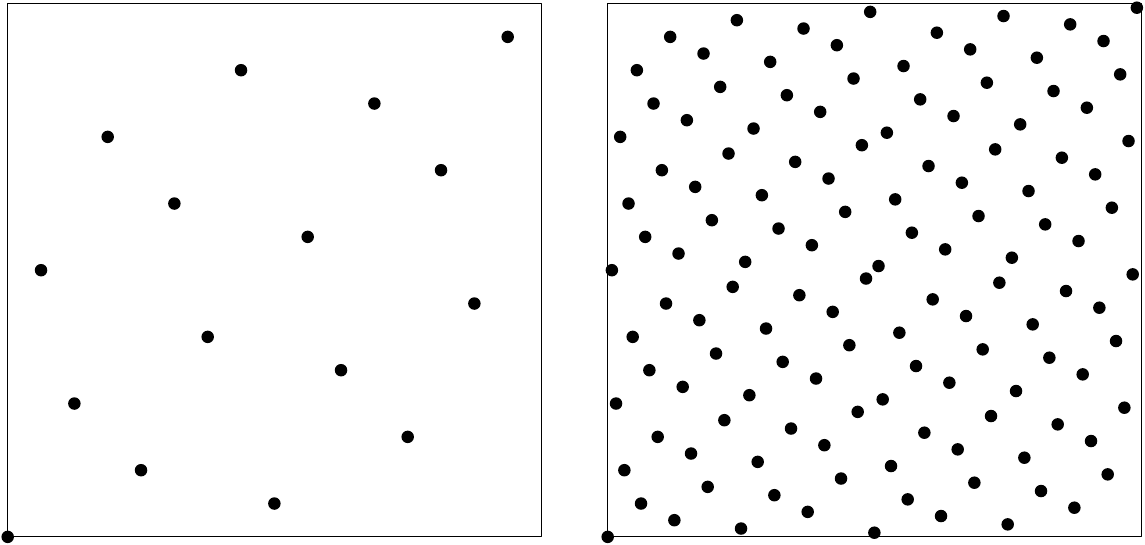}
    \caption{The Van der Corput set with $n = 16$ (left) and $n = 128$
       (right).}
    \figlab{vdc}
\end{figure}

\begin{lemma}
    \lemlab{0-net-box} For a parameter $\eps \in (0,1)$,there is a
    collection of $O(1/\eps)$ points $\PS \subset \HCX{2}$ such that
    any axis-aligned box $B$ with $\volX{B \cap \HCX{2}} \geq \eps$
    contains a point of $\PS$.
\end{lemma}
\begin{proof}
    Let $n = \ceil{4/\eps}$. We claim that the Van der Corput set of
    size $n$ is the desired point set $\PS$.

    Let $B$ be a box contained in $\HCX{2}$ of width $w$ and height
    $h$, with $wh \geq \eps$. Let $q \geq 2$ be the smallest integer
    such that $1/2^q < h/2 \leq 1/2^{q-1}$. By the choice of $q$, the
    projection of $B$ onto the $y$-axis contains an interval of the
    form $I = [k/2^q, (k+1)/2^q)$ for some integer $k$. Let
    $B_I = B \cap \Set{(x,y) \in \HCX{2}}{y \in I}$ be the box
    restricted to $I$ along the $y$-axis. Observe that
    \begin{align*}
      \volX{B_I} = w/2^q = w/(4 \cdot 2^{q-2}) \geq wh/4 \geq \eps/4
      \iff
      w \geq 2^q \eps/4.
    \end{align*}

    Let $S = [0,1] \times I$, so that each $\pp_j \in \PS \cap S$ has
    $\br(j) \in I$. In particular, the first $q$ binary digits of
    $\br(j)$ are fixed. This implies that the $q$ least significant
    binary digits of $j$ are fixed. In other words, $\PS \cap S$
    contains all points $p_j$ such that $j \equiv \ell \pmod{2^q}$ for
    some integer $\ell$---the $x$-coordinates of the points in $\PS$
    are regularly spaced in the strip $S$ with distance $2^q/n$. If
    the width of $B_I$ is at least $2^q/n$, then this implies that $B$
    contains a point of $\PS$ in the strip $S$.  Indeed, by the choice
    of $n$, $2^q/n \leq 2^q \eps/4 \leq w$.
\end{proof}

By extending the definition of the Van der Corput set to higher
dimensions, the above proof also generalizes.

\begin{definition}[the Halton-Hammersely set]
    \deflab{hh}%
    For a prime number $\rho$ and an integer
    $\alpha = \sum_{i=0}^\infty \rho^i b_i$, with
    $b_i \in \set{0, \ldots, \rho-1}$, written in base $\rho$, define
    $\br_\rho(\alpha) = \sum_{i=0}^\infty b_i/\rho^{i+1}$.  Note that
    $\br_2 = \br$ from \defref{vdc}.

    For integers $n$ and $d$, the \emphi{Halton-Hammersely set} is the
    collection of points
    \begin{equation*}
        p_1, \ldots, p_{n-1},
    \end{equation*}
    where
    $p_i = (\br_{\rho_1}(i), \br_{\rho_2}(i), \ldots,
    \br_{\rho_{d-1}}(i), i/n)$, and $\rho_1, \ldots, \rho_{d-1}$ are
    the first $d-1$ prime numbers. (Making $i/n$ the $d$\th coordinate
    instead of the 1st coordinate simplifies future notation.)
\end{definition}

\begin{lemma}
    \lemlab{0-net-box-d}%
    For a parameter $\eps \in (0,1)$, there is a collection of
    $\displaystyle {2^{O(d\log d)}}/{\eps}$ points
    $\PS \subset \HCX{d}$ such that any axis-aligned box $B$ with
    $\volX{\bigl.\smash{B \cap \HCX{d}}} \geq \eps$ contains a point
    of $\PS$.
\end{lemma}
\begin{proof}
    The proof is similar to \lemref{0-net-box}, with the Chinese
    remainder theorem as the additional tool.

    Let $n = \ceil{\smash{(2^{d-1}/\eps)\cdot(d-1)\sharp}\bigr.}$,
    where $k\sharp$ is the \emphi{primorial} function, defined as the
    product of the first $k$ prime numbers. It is known that
    $k\sharp \leq \exp\pth{(1 + o(1))k\log k}$, which implies
    $n = 2^{O(d\log d)}/\eps$.  We claim that the Halton-Hammersely
    set of size $n$ is the desired point set $\PS$.

    Denote the side lengths of the box $B$ by $s_1, \ldots, s_d$, with
    $\prod_{i=1}^d s_i \geq \eps$. For each $i = 1, \ldots, d-1$, let
    $q_i$ be the smallest integer such that
    $1/\rho_i^{q_i} < s_i/2 \leq 1/\rho_i^{q_i-1}$, where $\rho_i$ is
    the $i$\th prime number. By the choice of $q_i$, the projection of
    $B$ onto the $i$\th axis contains an interval of the form
    $I_i = [k_i/\rho_i^{q_i}, (k_i+1)/\rho_i^{q_i}]$ for some integer
    $k_i$.  Let $S$ denote the box
    $I_1 \times \ldots \times I_{d-1} \times [0,1]$ and
    $B_S = B \cap S$. Observe that
    \begin{align*}
      \volX{B_S}
      = s_d \prod_{i=1}^{d-1} \frac{1}{\rho_i^{q_i}}
      \geq s_d \prod_{i=1}^{d-1} \frac{s_i}{2\rho_i}
      \geq \frac{\eps}{2^{d-1}} \prod_{i=1}^{d-1} \frac{1}{\rho_i}
      \iff
      s_d \geq \frac{\eps}{2^{d-1}} \prod_{i=1}^{d-1} \rho_i^{q_i-1}.
    \end{align*}

    Similar to \lemref{0-net-box}, we observe that the point
    $p_j \in P$ falls into $S$ when
    $j \equiv \ell_i \pmod{\rho_i^{q_i}}$ for some integers
    $\ell_1, \ldots, \ell_{d-1}$. By the Chinese remainder theorem,
    there is exactly one number in the set
    $\set{0, 1, \ldots, \prod_{i=1}^{d-1} \rho_i^{q_i}-1}$ (the $d$\th
    coordinate of $p_j$) which satisfies these $d-1$ equations. In
    particular, the points in $\PS \cap S$ are spaced regularly along
    the $d$\th axis with distance
    $\delta = (1/n)\prod_{i=1}^{d-1} \rho_i^{q_i}$.  Once again, we
    argue that the length of $B$ along the $d$\th axis is at least
    $\delta$, which implies the result.  Indeed, by our choice of $n$
    we have that,
    \begin{align*}
      \delta = \frac{1}{n}\prod_{i=1}^{d-1} \rho_i^{q_i}
      \leq
      \frac{\eps}{2^{d-1}} \prod_{i=1}^{d-1} \rho_i^{q_i-1}
      \leq s_d.
    \end{align*}
\end{proof}

\end{document}